\pdfoutput=1
\documentclass[sigconf, 10pt, screen, table]{acmart}
\renewcommand\footnotetextcopyrightpermission[1]{}

\usepackage{amsmath}
\usepackage{xcolor}
\usepackage{multirow}
\usepackage{array}
\usepackage{subcaption}
\usepackage[noline, noend, linesnumbered]{algorithm2e}
\usepackage{booktabs}
\usepackage{amsfonts}
\usepackage{cleveref}
\usepackage{tikz}
\usepackage{algpseudocode}
\usepackage{listings}
\usepackage{multicol}
\usepackage{bbm}

\usepackage[all]{background}
\SetBgContents{\color{pink}{Under anonymous submission}}
\SetBgPosition{8.9cm,-5.5cm}
\SetBgOpacity{1.0}
\SetBgAngle{270.0}
\SetBgScale{2.0}

\graphicspath{{images/}}

\crefformat{section}{§#2#1#3}

\newtheorem{assumption}{Assumption}
\newtheorem{lemma}{Lemma}
\newtheorem{theorem}{Theorem}
\newcommand{\norm}[1]{\left\lVert#1\right\rVert}

\newenvironment{alg}[1][!t]
{
  \begin{algorithm}[#1]%
}{\end{algorithm}}
\algrenewcommand{\algorithmiccomment}[1]{\hfill\hfill\hfill\hfill\hfill $\triangleright$ #1}
\newcommand{\algrule}[1][.7pt]{\par\vskip.5\baselineskip\hrule height #1\par\vskip.5\baselineskip}

\usepackage{url}

\definecolor{mygreen}{rgb}{0,0.6,0}
\definecolor{mygray}{rgb}{0.5,0.5,0.5}
\definecolor{mymauve}{rgb}{0.58,0,0.82}

\newcommand{\PHB}[1]{\noindent\textbf{#1}\hspace{.5em}} 
\newcommand{\PHM}[1]{\vspace{.2em}\noindent\textbf{#1}\hspace{.5em}} 

\begin{document}
\hypersetup{linkcolor=blue, citecolor=magenta}

\title{Pisces: Efficient Federated Learning \\via Guided Asynchronous Training}


\author{Zhifeng Jiang}
\affiliation{%
  \institution{HKUST}}
\email{zjiangaj@cse.ust.hk}

\author{Wei Wang}
\affiliation{%
  \institution{HKUST}}
\email{weiwa@cse.ust.hk}

\author{Baochun Li}
\affiliation{%
  \institution{University of Toronto}}
\email{bli@ece.toronto.edu}

\author{Bo Li}
\affiliation{%
  \institution{HKUST}}
\email{bli@cse.ust.hk}








\begin{abstract}
    Federated learning (FL) is typically performed in a synchronous parallel manner, where the involvement of a slow client delays a training iteration. Current FL systems employ a participant selection strategy to select fast clients with quality data in each iteration. However, this is not always possible in practice, and the selection strategy often has to navigate an unpleasant trade-off between the speed and the data quality of clients. 
    
    In this paper, we present Pisces, an \emph{asynchronous} FL system with intelligent participant selection and model aggregation for accelerated training. To avoid incurring excessive resource cost and stale training computation, Pisces uses a novel scoring mechanism to identify suitable clients to participate in a training iteration. It also adapts the pace of model aggregation to dynamically bound the progress gap between the selected clients and the server, with a provable convergence guarantee in a smooth non-convex setting.
    We have implemented Pisces in an open-source FL platform called Plato, and evaluated its performance in large-scale experiments with popular vision and language models. Pisces outperforms the state-of-the-art synchronous and asynchronous schemes, accelerating the time-to-accuracy by up to $2.0\times$ and $1.9\times$, respectively.
\end{abstract}

\maketitle

\pagestyle{plain}

\section{Introduction}
\label{sec:intro}

Federated learning~\cite{mcmahan2017communication} (FL) enables clients to
collaboratively train a model in a privacy-preserving manner under the
orchestration of a central server. At its core, FL keeps private data
decentralized, while only allowing clients to share local updates that
contain minimum information, such as the gradients or the model weights~\cite
{kairouz2019advances}. By evading the privacy risks of centralized learning,
FL has gained popularity in a multitude of applications, such as mobile
services~\cite{yang2018applied, hard2018federated, ramaswamy2019federated,
chen2019federated, google2021assistant, paulik2021federated}, financial
business~\cite{ludwig2020ibm, webank2020laundering}, and medical care \cite
{li2019privacy, nvidia2020oxygen}.

Current FL systems orchestrate the training process with a \textit
{synchronous parallel scheme}, where the server waits for all the
participating clients to finish local training before updating the global
model with the aggregated updates~\cite
{mcmahan2017communication, bonawitz2019towards}. While synchronous FL is easy
to implement, the \textit{time-to-accuracy}, measured by the wall clock time
needed to train a model to the target accuracy, can be substantially delayed
in the presence of \textit{stragglers} whose updates arrive much later than
others. The straggler problem is exacerbated when the clients have
order-of-magnitude difference in speed, which is commonplace in 
cross-device scenarios \cite{wu2019machine, yang2021characterizing,
lai2021oort} (\cref{sec:motivation_sync}).

A common approach to mitigating stragglers is to identify slow clients and exclude them from the participants in the current iteration. Existing works propose various metrics to quantify the clients' computing capabilities, based on which they reduce the involvement of slow clients~\cite[]{nishio2019client, chai2020tifl, lai2021oort}. Yet, participant selection may not always work well. Consider a pathological scenario where the clients' speeds and data quality are \textit{inversely correlated}~\cite{lai2021oort, huba2021papaya}. In this case, prioritizing fast clients inevitably overlooks those clients who are slow yet informative with quality data. As the time-to-accuracy depends on both participants' speeds and data quality, reconciling the demand for the two yields a \textit{knotty trade-off}. We show in \cref{sec:motivation_limit} that even Oort~\cite{lai2021oort}, the state-of-the-art participant selection strategy, can sometimes underperform random selection by 2.7$\times$ when navigating the trade-off.

In this paper, we advocate a more radical solution, that is, to switch to an \textit{asynchronous} FL design. In asynchronous FL, the server can both (1) select a subset of idle clients to run and (2) aggregate received local updates at any time, regardless of whether some clients are still in progress. As such, we can avoid blocking on straggling clients, thus relieving the pathological tension between clients' speeds and data quality. While the intuition is simple, there remain a few design challenges for efficient asynchronous FL.

First, being asynchronous results in more frequent client invocations, which may
harm the \textit{resource efficiency}. Existing approaches commonly
involve all clients and keep them running throughout the training
process~\cite{xie2020asynchronous, chen2020asynchronous, shi2020hysync}.
Compared with selecting only a small portion (e.g., 25\%) of clients, keeping
everyone busy yields marginal performance gains (e.g., $\leq$1.5$\times$) with
substantial resource overhead (e.g., 3.6-6.7$\times$). While this points to
controlled concurrency~\cite{nguyen2021federated}, it remains \textit
{unexplored} how to fully utilize the concurrency quota with optimized participant
selection for asynchronous FL.

Second, asynchrony encourages more intensive aggregation, which risks inducing \textit{stale computation}. When a client reports its local update, the global model has probably proceeded ahead of the initial version that the client bases on, to how many steps are termed \textit{staleness} of the client and affects the quality of the update. As incorporating local updates with larger staleness makes the global model converge slower both theorectically~\cite{nguyen2021federated} and empirically~\cite{xie2020asynchronous}, it is desirable to bound clients' staleness to a low level. While existing approaches such as buffered aggregation~\cite{nguyen2021federated} are somehow effective in lowering clients' staleness, they do not guarantee a bound in theory and rely on manual tuning efforts to adapt to different federation environments.

In this paper, we present Pisces, an end-to-end asynchronous parallel scheme for efficient FL training (\cref{sec:overview}). To make good use of each quota under controlled concurrency, Pisces conducts \textit{guided participant selection}. In a nutshell, Pisces prioritizes clients with high data quality, which is measured by clients' training loss based on the approximation of importance sampling~\cite{johnson2018training, katharopoulos2018not}. Given that the training loss may be misleading as a result of corrupted data or malicious clients, Pisces proposes to cluster clients' loss values to identify outliers for preclusion from training. Moreover, Pisces uses the predictability of clients' staleness to avoid inducing stale computation. Our design eliminates the unpleasant trade-off between clients' speeds and data quality, reaping the most resource efficiency for being asynchronous (\cref{sec:selection}).

To avoid being arbitrarily impacted by stale computation that is already induced, Pisces further adopts an \textit{adaptive aggregation pace control} with a novel online algorithm. By dynamically adjusting the aggregation interval to match the speeds of running clients, Pisces balances the aggregation load over time for both being steady in convergence and scalable in practice. The algorithm automatically adapts to different distributions of clients' speed without manual tuning (\cref{sec:aggregation}). It can also bound the clients' staleness under any target value, with a provable convergence guarantee in a smooth non-convex setting (\cref{sec:convergence}).

We have implemented Pisces atop Plato~\cite{tlsystem2021plato}, an open-source FL platform (\cref{sec:implementation}), under realistic settings of system and data heterogeneity across 100-400 clients (\cref{sec:evaluation}). Extensive experiments over image classification and language model applications show that, compared to Oort, the cutting-edge synchronous FL solution, and FedBuff \cite{nguyen2021federated}, the state-of-the-art asynchronous FL design, Pisces achieves the time-to-accuracy speedup by up to $2.0\times$ and $1.9\times$, respectively. Pisces is also shown to be insensitive to the choice of hyperparameters. Pisces will be open-sourced.

In summary, we make the following contributions in this work:

\begin{enumerate}
    \item We highlight the knotty trade-off between clients' speeds and data quality in synchronous FL.
    \item We propose algorithms to automate participant selection and model aggregation in asynchronous FL for dealing with resource efficiency and stale computation.
    \item We implement and evaluate these algorithms in Pisces and show the improvements over the state-of-the-art under realistic settings.
\end{enumerate}
\section{Background and Motivation}\label{sec:motivation}

In this section, we briefly introduce synchronous FL and its inefficiency in
the presence of stragglers (\cref{sec:motivation_sync}). We then discuss the
limitations of current participant selection strategies when navigating the trade-off of clients' speeds and data quality in synchronous FL, thus motivating the need for relaxing the synchronization barrier (\cref{sec:motivation_limit}).

\subsection{Synchronous Federated Learning}\label{sec:motivation_sync}

Federated learning (FL) has become an emerging approach for building a model from decentralized data. To orchestrate the training process, current FL systems employ the parameter server architecture \cite{smola2010architecture, li2014scaling} with a \textit{synchronous parallel} design \cite{mcmahan2017communication, bonawitz2019towards}, where a server refines the global model on a round basis, as illustrated in Fig.~\ref{fig:fl_sync}. In each round, the server randomly selects a subset of online clients to participate and sends the current global model to them. These clients then perform a certain number of training steps using their local datasets. Finally, the server waits for all participants to report their local updates and uses the aggregated updates to refine the global model before advancing to the next round.

\begin{figure}[t]
  \centering
  \begin{subfigure}[b]{0.49\columnwidth}
    \centering
    \includegraphics[width=\columnwidth]{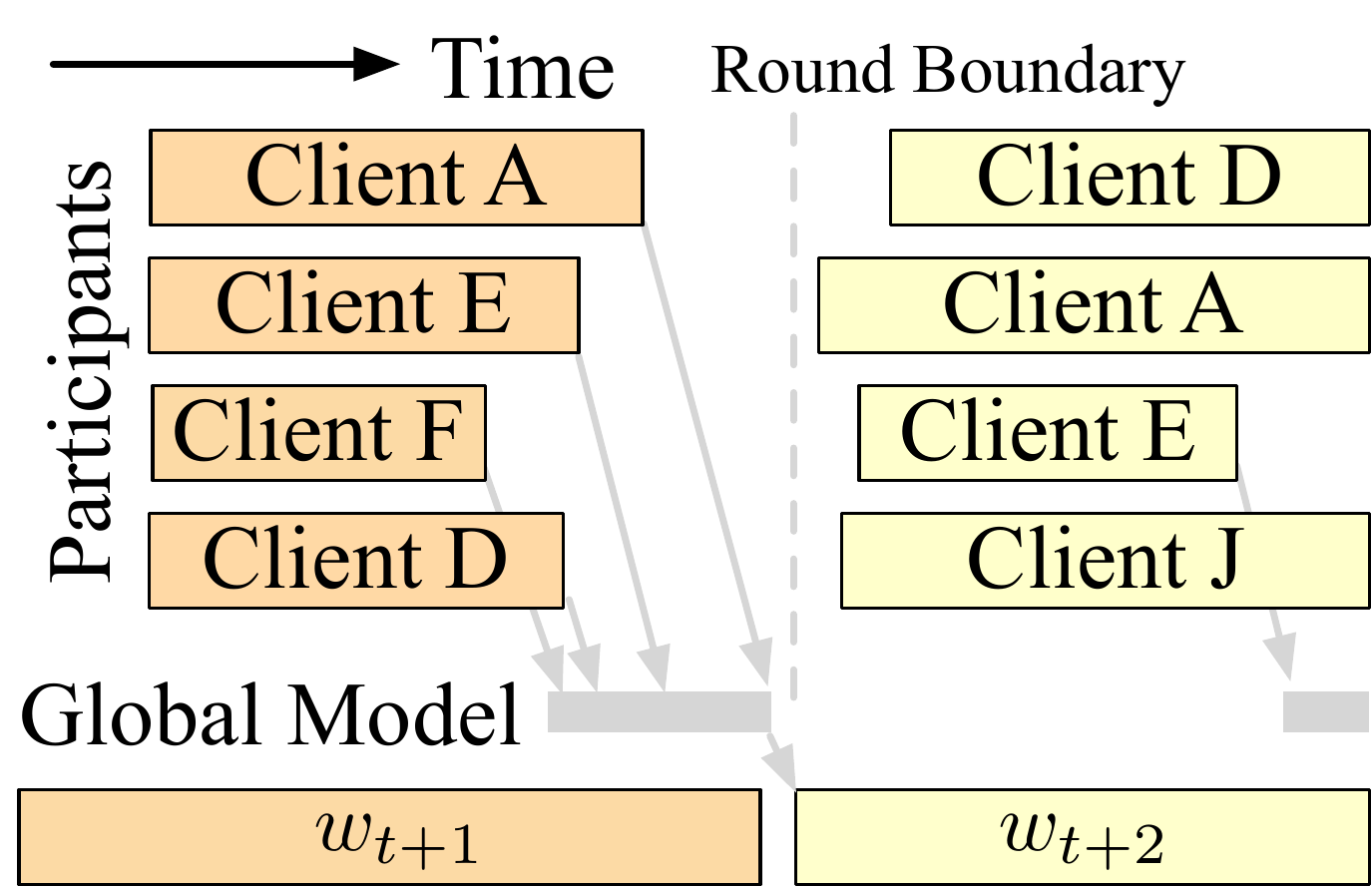}
    \caption{Synchronous}
    \label{fig:fl_sync}
\end{subfigure} \hfill
\begin{subfigure}[b]{0.49\columnwidth}
  \centering
  \includegraphics[width=\columnwidth]{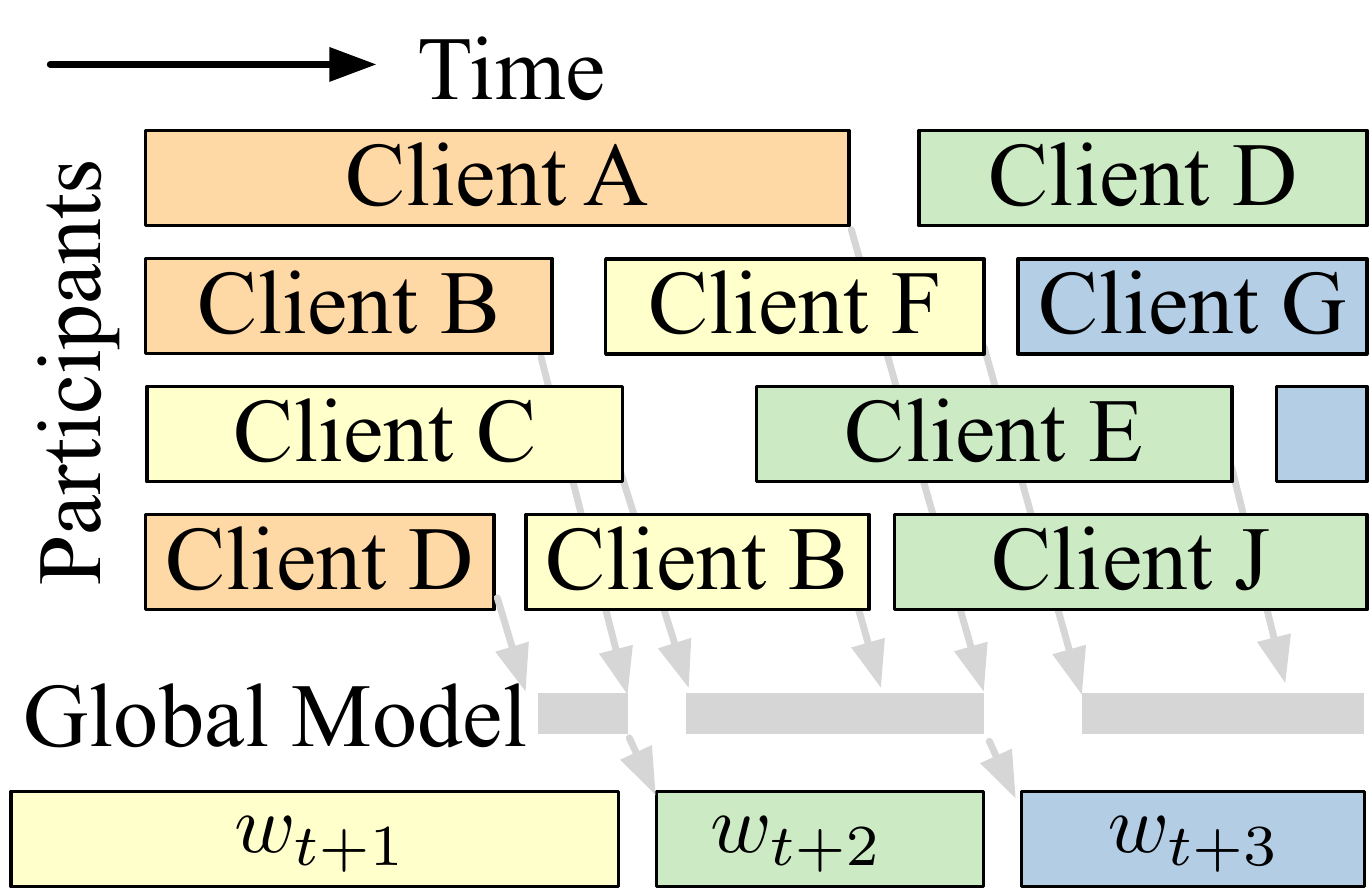}
  \caption{Asynchronous}
  \label{fig:fl_async}
\end{subfigure}
\caption{Sychronous and asynchronous FL \cite{huba2021papaya}.}
\label{fig:fl}
\end{figure}

\PHM{Performance bottlenecks.} The performance of synchronous FL can be significantly harmed by \textit{straggling clients} that process much slower than non-stragglers. This problem becomes particularly salient in cross-device scenarios, where the computing power and data amount vary across clients by orders of magnitude~\cite{caldas2018leaf, wu2019machine, yang2021characterizing, lai2021oort}. In our testbed experiments (detailed in Sec.~\ref{sec:evaluation_method}), the server running FedAvg~\cite{mcmahan2017communication} algorithm remains idle for 33.2-57.2\% of the training time waiting for the slowest clients to report updates.

To mitigate stragglers, simple solutions include periodic aggregation or over-selection \cite{bonawitz2019towards}. The former imposes a deadline for participants to report updates and ignores late submissions; the latter selects more participants than needed but only waits for the reports from a number of early arrivals. Both solutions waste the computing efforts of slow clients, leading to suboptimal resource efficiency. Developers then seek to improve over random selection by identifying and excluding stragglers in the first place.

\subsection{Participant Selection}\label{sec:motivation_limit}

By prioritizing fast clients, the average round latency will be shortened compared to random selection. However, this is not sufficient to achieve a shorter \textit{time-to-accuracy}, which is the product of average round latency and the number of rounds taken to reach the target accuracy. To avoid inflating the number of rounds when handling stragglers, participant selection should also account for the clients' data quality. As data quality may not be positively correlated with speeds, a good strategy must strike a balance in-between.

\PHM{Prior arts.} To navigate the trade-off between the two factors, extensive research efforts such as FedCS~\cite{nishio2019client}, TiFL~\cite{chai2020tifl}, Oort~\cite{lai2021oort} and AutoFL~\cite{kim2021autofl} have been made in the literature, among which Oort is the state-of-the-art for its fine-grained navigation and training-free nature. To guide participant selection, Oort characterizes each client $i$ with a continuous utility score $U^{Oort}_i$ that jointly considers the client's speed and data quality:

\begin{equation}
  U^{Oort}_i = \underbrace{|B_i| \sqrt{\frac{1}{|B_i|}\sum_{k \in B_i} Loss(k) ^2}}_{Data \  quality} \times \underbrace{(\frac{T}{t_i})^{\mathbbm{1}(T<t_i) \times \alpha}}_{System \ speed}.
  \label{eq:oort_utility}
\end{equation}

In a nutshell, the first component is the aggregate training loss that reflects the volume and distribution of the client's dataset $B_i$. The second component compares the client's completion time $t_i$ with the developer-specified duration $T$ and penalizes any delay (which makes the indicator $\mathbbm{1}(T<t_i)$ outputs $1$) with an exponent $\alpha > 0$.

\PHM{Inefficiency.} We briefly explain the \textit{strict penalty effect} that Eq.~\eqref{eq:oort_utility} imposes on slow clients. Following the evaluation setting in Oort, assume $\alpha=2$. The quantified data quality of a straggler $i$ will then be divided by a factor proportional to the square of its latency $t_i$. Given that in Oort a client is selected with probability in proportion to its utility score, such a penalty implies that this straggler has much less chance of being selected in the training than non-stragglers.

However, imposing such a heavy penalty is not always desirable. Consider a pathological case where the clients' speeds and data quality are inversely correlated, i.e., faster clients are coupled with fewer data of poorer quality. Note that this is not uncommon in practice~\cite{lai2021oort, huba2021papaya}; for example, it usually takes a client with a larger dataset a longer time to finish training. In this case, strictly penalizing slow clients can lead to using an insufficient amount or quality of data, which can impair the time-to-accuracy compared to no optimization. To illustrate this problem, we compare Oort over FedAvg~\cite{mcmahan2017communication} (that uses random selection) in a small-scale training task where 5 out of 20 clients are selected at each round to train over the MNIST dataset. In this emulation experiment (detailed settings in Sec.~\cref{sec:evaluation_method}), the completion time of clients are configured following the Zipf distribution ($a = 1.2$) \cite{lee2018pretzel, tian2021crystalperf, jia2021boki} so that the majority are fast while the rest are extremely slow. Accordingly, fast clients are associated with fewer data samples of more unbalanced label distribution and vice versa, leveraging latent Dirichlet allocation (LDA)~\cite{hsu2019measuring, reddi2020adaptive, al2020federated, acar2021federated}.

\begin{figure}[t]
  \centering
  \begin{subfigure}[b]{0.54\columnwidth}
    \centering
    \includegraphics[width=\columnwidth]{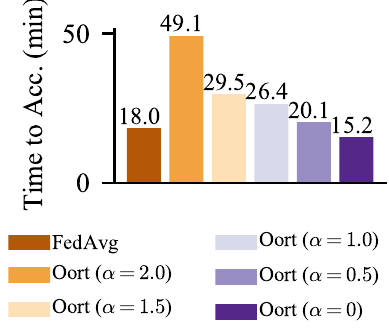}
    \caption{MNIST Model}
    \label{fig:inverse_mnist}
  \end{subfigure}
  \begin{subfigure}[b]{0.42\columnwidth}
    \centering
    \includegraphics[width=\columnwidth]{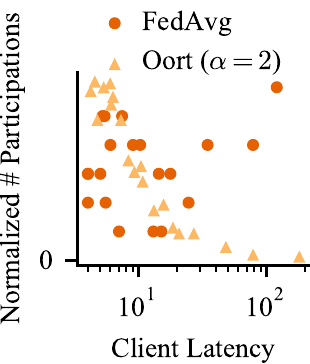}
    \caption{MNIST Clients}
    \label{fig:inverse_mnist_2}
  \end{subfigure}
  \caption{Slow clients are overlooked by Oort \cite{lai2021oort} when clients' speed and data quality are at odds.}
  \label{fig:inverse}
\end{figure}

Fig.~\ref{fig:inverse_mnist} depicts the time taken to reach 95\% accuracy. Oort with straggler penalty factor $\alpha=2.0$ suffers $2.7\times$ slowdown than FedAvg. According to Fig.~\ref{fig:inverse_mnist_2}, Oort's poor performance results from its bias toward fast clients. Under strict penalty, stragglers in Oort get way less attention than others, despite having rich data of good quality. Instead, FedAvg evenly picks each client, which involves stragglers for enough times and thus leads to faster convergence.

\PHM{Generality.} To confirm that the strict penalty effect generally exists, we evaluate Oort across different small $\alpha$'s: 2.0, 1.5, 1.0, 0.5, and 0. As Fig.~\ref{fig:inverse_mnist} shows, while using a more gentle straggler penalty factor (i.e., smaller $\alpha$) does yield a shorter time-to-accuracy, the use of nonzero factors is still harmful. For Oort to improve over FedAvg, it should ignore the speed disparity and purely focus on prioritizing clients with high data quality (i.e., $\alpha = 0$). This strategy, however, deviates from Oort's original design and mandates manual tuning with prior knowledge. This limitation also generally applies to other optimization alternatives due to the tricky trade-off between clients' speeds and data quality.

In short, participant selection does not completely address the performance bottlenecks in synchronous FL. The limited tolerance for stragglers is responsible for such inefficiency.
\section{Pisces Overview}\label{sec:overview}

Pisces improves training efficiency by performing asynchronous FL with novel algorithmic principles. In this section, we present the configuration space and system designs to help the reader follow the principles that are later introduced.

\PHB{Design knobs.} The advantages of switching to an asynchronous design are two-fold. \textit{First}, it can inform an available client to train whenever it is idle, as illustrated in Fig.~\ref{fig:fl_async}. As such, fast clients do not have to idle waiting for stragglers to finish their tasks. \textit{Second}, the server can conduct model aggregation as soon as a local update becomes available. Thus, the pace at which the global model evolves can be lifted by fast clients, instead of being constrained by stragglers.

Putting it altogether, Pisces aims to unleash the performance potential of FL training by answering the following two sets of questions:

\begin{enumerate}
  \item How many clients are selected to run concurrently? If there is an available quota in an iteration, how to decide whether to launch local training, and which clients to select (\cref{sec:selection})?
  \item In each iteration, determine whether the received model updates need aggregation and how should they be aggregated (\cref{sec:aggregation}).
\end{enumerate}

\begin{figure}[t]
    \centering
    \includegraphics[width=0.80\columnwidth]{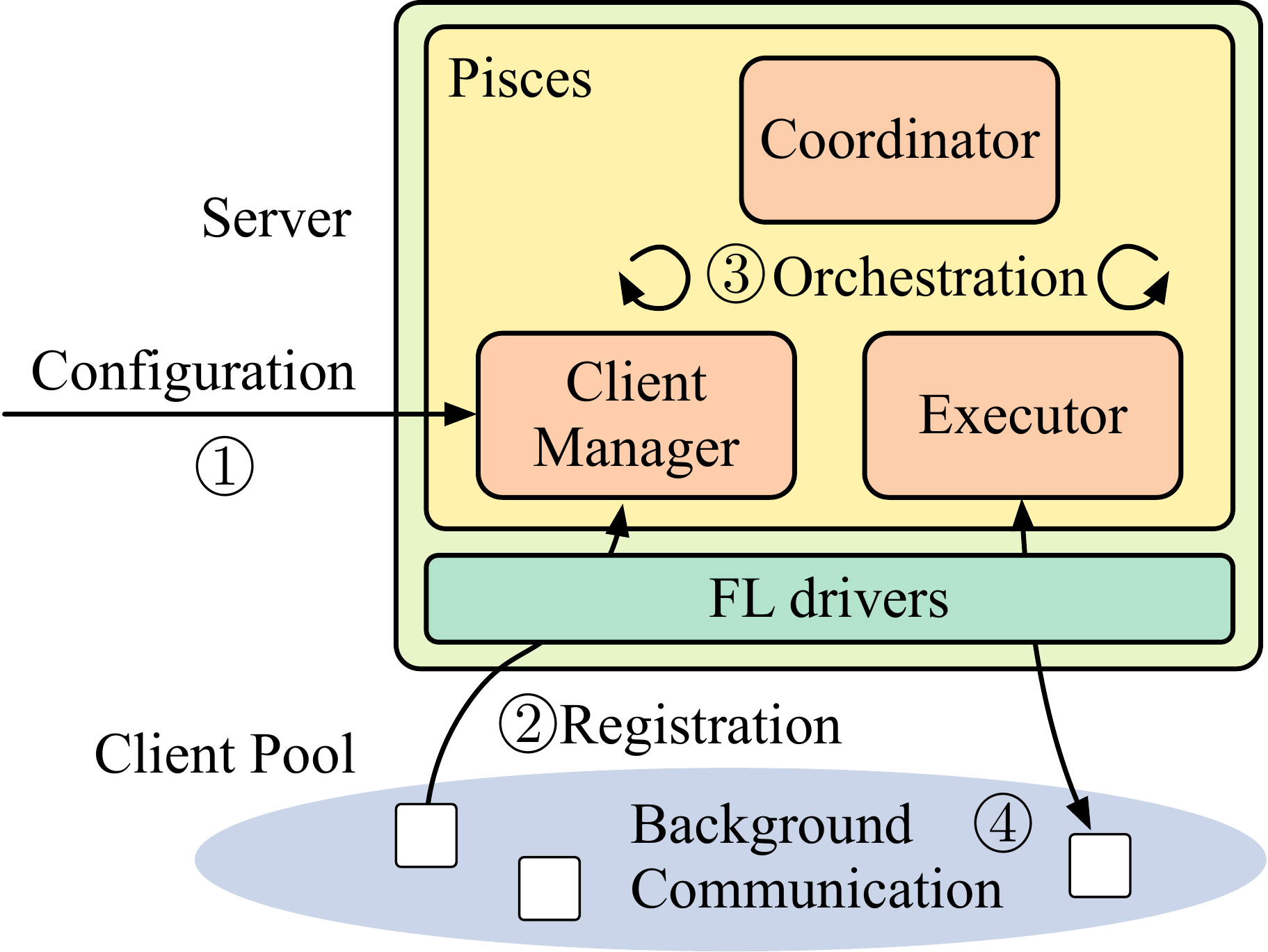}
    \caption{Pisces' architecture.}
    \label{fig:Pisces}
\end{figure}

\PHM{Architecture.} Given the client selection and model aggregation strategies, Pisces implements them in the training workflow by orchestrating three components: client manager, coordinator, and executor. Fig.~\ref{fig:Pisces} depicts an architecture overview. \textcircled{\raisebox{-1.0pt}{1}} \textit{Configuration}: given a training plan, the client manager configures itself and waits for clients to arrive. \textcircled{\raisebox{-1.0pt}{2}} \textit{Registration}: upon a client's arrival, the client manager registers its meta-information (e.g., dataset size) and starts to monitor its runtime metrics (e.g., response latency). \textcircled{\raisebox{-1.0pt}{3}} \textit{Orchestration}: during training, the coordinator iterates over a control loop that interacts with both the client manager and executor. \textcircled{\raisebox{-1.0pt}{4}} \textit{Background Communication}: in the background, the executor handles the server-client communication, maintains a buffer that stores non-aggregated local updates, and validates the model with hold-out datasets.

\begin{figure}[t]
    \begin{lstlisting}[belowskip=-2.0 \baselineskip, escapechar=|]
  import executor, client_manager
  
  def asynchronous_training():
    # Repeat every time window
    while True:
      # Perform model aggregation if necessary
      if client_manager.need_to_aggregate():|\label{line:loop_agg}|
        executor.aggregate()
          
      # Exit the loop if applicable
      if executor.to_terminate():
        break|\label{line:loop_exit}|
      
      # Select participants if necessary
      if client_manager.need_to_select():|\label{line:loop_slt}|
        clients = client_manager.select_clients()
        executor.start_local_training(clients)\end{lstlisting}
    \caption{Pseudocode of the coordinator's control loop.}
    \label{fig:loop}
\end{figure}

We further zoom in on the coordinator's control loop, as shown in Fig.~\ref{fig:loop}. At each iteration, the coordinator first asks the client manager whether to perform model aggregation (Line~\ref{line:loop_agg}). If necessary, it delegates the task to the executor for completion. The coordinator then consults the client manager on whether any idle client needs to be selected (Line~\ref{line:loop_slt}). The client manager will either reply no, or yes with a plan to instruct the executor on whom to select.

In the next two sections, we focus on Pisces' algorithmic principles designed to tackle the concerns of resource efficiency and stale computation arising from asynchrony.
\section{Participant Selection}\label{sec:selection}

In this section, we consolidate the need for controlling the maximum number of clients allowed to run concurrently (i.e., concurrency) (\cref{sec:selection_concurrency}), and then introduce how Pisces selects clients for fully utilizing the concurrency quotas (\cref{sec:selection_utility}).

\subsection{Controlled Concurrency}\label{sec:selection_concurrency}

Idle clients can be invoked at any time in asynchronous FL. Having all clients training concurrently, however, can \textit{saturate} the server's memory space and network bandwidth. Even with abundant resources, it remains a question whether an excessive resource usage can translate into a \textit{proportional} performance gain. Many existing works continuously invoke all clients for maximizing the speedup~\cite{xie2020asynchronous, chen2020asynchronous, shi2020hysync}. There exists a work, i.e., FedBuff~\cite{nguyen2021federated}, that evaluates various concurrencies; however, it neither points out the scaling issue of concurrency in asynchronous FL nor does it consider how to fully utilize the concurrency quotas. How to maximize the \textit{resource efficiency} is thus a challenge in designing the principles of participant selection.

\PHB{Is selection still necessary?} To empirically examine the trade-off between resource cost and run-time gain, we characterize the scaling behavior of FedBuff using the same 20-client testbed mentioned in Sec.~\ref{sec:motivation_limit}. We regard it as a representative because it is the state-of-the-art asynchronous FL approach deployed in Facebook AI \cite{huba2021papaya}. Essentially, it selects clients randomly and employs buffered aggregation where model aggregation is conducted after the number of received local updates exceeds a certain aggregation goal (more details in Sec.~\ref{sec:aggregation_staleness}). Here we consistently set the aggregation goal to be 40\% of the concurrency limit.

\begin{figure}[t]
    \centering
  \begin{subfigure}[b]{0.42\columnwidth}
    \centering
    \includegraphics[width=\columnwidth]{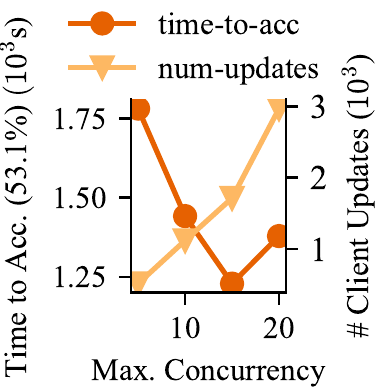}
    \caption{FEMNIST (Image)}
    \label{fig:async_scale_femnist}
  \end{subfigure}
  \hspace{15px}
  \begin{subfigure}[b]{0.45\columnwidth}
    \centering
    \includegraphics[width=\columnwidth]{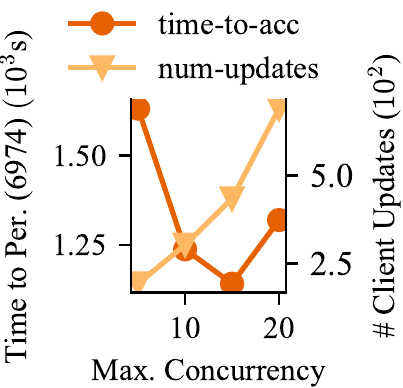}
    \caption{StackOverflow (LM)}
    \label{fig:async_scale_stackoverflow}
  \end{subfigure}
  \caption{Scaling up the concurrency leads to diminishing gain in run time with 
  escalating resource usage.}
  \label{fig:async_scale}
  \vspace{-0.3cm}
\end{figure}

\definecolor{my-red}{RGB}{230,97,1}
\definecolor{my-orange}{RGB}{253,184,99}

In Fig.~\ref{fig:async_scale}, we can see a reduction in time-to-accuracy as the maximum allowed concurrency grows (\textcolor{my-red}{dark lines}), though at a \textit{diminishing} speed. After a \textit{turning} point (around 15 clients), the time-to-accuracy instead starts to inflate. On the other hand, we constantly observe a \textit{superlinear} increase in the accumulated number of clients' updates when scaling up the maximum concurrency (\textcolor{my-orange}{light lines}). Such a growing pattern indicates an escalating burden on network bandwidth and limitations in scalability. Hence, unbounded concurrency impairs resource efficiency, and it is necessary to limit the concurrency to a small portion of the population.

\subsection{Utility Profiling}\label{sec:selection_utility}

With limited concurrency quotas, the challenge of maximizing resource efficiency boils down to how to select participants given each available quota.

\PHM{Does random selection suffice?} As stragglers can be tolerated, it naturally raises the question of whether random selection suffices for asynchronous FL before resorting to more complex methods. Our answer is no, as clients widely exhibit heterogeneous data distributions~\cite{hsieh2020non}. Even if clients have the same speed, it is still desirable to focus on clients with more quality data for taking a larger step towards convergence each time an update is incorporated. We empirically show the inferiority of random selection in Sec.~\ref{sec:evaluation_end2end}.

\PHM{Which clients to prioritize?} We have left off the question of how to identify clients with the most \textit{utility} to improve model quality. Starting from the state-of-the-art solution established for synchronous FL (\cref{sec:motivation_limit}), we need to address the following concerns arising from asynchronous FL:

\begin{itemize}
    \item Given lifted tolerance for stragglers, do they still have to be strictly penalized in the chance of being selected?
    \item Given new training dynamics, does clients' data quality influence the global model in the same way?
\end{itemize}

Both of the answers are no. First, an individual client's speed \textit{does not determine the pace at which the global model updates}. With the removal of synchronization barriers, there is no need to wait for stragglers to finish throughout the training process. On the other hand, strictly penalizing stragglers risks precluding informative clients given the coupled nature of speeds and data quality (~\cref{sec:motivation_limit}). Thus, getting rid of slow clients can do more harm than good.

On the other hand, a client's speed \textit{indirectly affects its utility to improve the global model}. The longer it takes a client to train, the more changes the global model that it bases on is likely to undergo, because of other clients' contributions in the interim. In the literature, it is dubbed as \textit{staleness} the lag between the version of the current global model and that of the locally used one. Empirical studies show that as the staleness of an update grows, the accuracy gain from incorporating that update will diminish \cite{ho2013more, cui2014exploiting, xie2020asynchronous}. Hence, it is not so useful to select clients with high-quality data but a large chance to produce stale updates.

Given the two insights, we formulate the  utility of a client by respecting the roles that its data quality and staleness play in improving the global model:

\begin{equation}
    U^{Pisces}_i = \underbrace{|B_i| \sqrt{\frac{1}{|B_i|}\sum_{k \in B_i} Loss(k) ^2}}_{Data \  quality} \times \underbrace{\frac{1}{(\tilde{\tau_i} + 1)^\beta}}_{Staleness},
    \label{eq:Pisces_utility}
\end{equation} where $Loss(k)$ is the training loss of sample $k$ from the client $i$'s local dataset $B_i$, $\tilde{\tau_i} \geq 0$ is the estimated staleness of the client's updates and $\beta > 0$ is the staleness penalty factor. Based on the profiled utilities, Pisces sorts clients and selects the clients with the highest utilities to train.

\PHM{Robustness against training loss outliers.} The first term of Eq.~\ref{eq:Pisces_utility} originates from Oort's utility formula (Eq.\ref{eq:oort_utility}) that approximates the ML principle of \textit{importance sampling} \cite{johnson2018training, katharopoulos2018not} to sketch a client's data quality. We do not reinvent the formulation as its effectiveness does not vary across synchronization modes. Moreover, it has the advantages of no computation overhead and negligible privacy leakage.

Beyond the formulation, however, we need to tackle a \textit{robustness challenge} unique to asynchronous FL. By definition, clients with high training loss are taken as possessing important data. In practice, a high loss could also result from the client's \textit{corrupted data} or \textit{malicious manipulation}. To quickly fix, Oort decreases the chance of the latter case by (i) adding randomness by probabilistic sampling and (ii) blacklisting clients who are selected over a threshold of times.

Unfortunately, both strategies do not generalize to asynchronous FL. \textit{First}, the server performs selection way more frequently,  eliminating the benefits of probabilistic sampling. Suppose that asynchronous FL selects clients every $5$ seconds and synchronous FL every $60$ second. For a client with a probability of 0.01 to be selected in an attempt, the probability that it must be selected within five minutes is $1 - (1-0.01)^5 \approx 5\%$ in synchronous FL, while being $1 - (1-0.01)^{60} \approx 45\%$ in asynchronous FL. \textit{Further}, as a client is generally involved more frequently, its participation can quickly reach the blacklist threshold. The client pool for selection can thus be exhausted in the late stage of training, hindering convergence.

Given the intuition that (i) the loss values of benign clients evolve in \textit{roughly the same direction}, while (ii) those resulting from corrupted or malicious clients tend to be outliers consistently, we propose to blacklist clients whose losses have been \textit{outliers over a threshold of time}. Initially, each client is given $r$ \textit{reliability credits}. For a client update which uses the global model of version $w_t$, Pisces pools the received client updates that are trained from similar initial versions of the global model, namely $\{w_{t-k}, w_{t-k+1}, \dots, w_{t}\}$ for some $k > 0$, and runs DBSCAN~\cite{ester1996density} to \textit{cluster} their associated losses. Each time a client's loss value is identified as an outlier, its credit gets deducted by one. A client will be removed from the client pool when it runs out of the reliability credits. As validated, Pisces can reduce anomalous updates while avoiding blindly blacklisting benign clients (\cref{sec:evaluation_sensitivity}).

\PHM{Staleness-aware discounting.} As the second term in Eq.~\ref{eq:Pisces_utility} shows, Pisces \textit{discounts} a client's utility based on its estimated staleness of its local updates. We adopt the reciprocal of a polynomial function for realizing two goals. (1) \textit{Functionality}: given the same data quality, the client with larger staleness should get a larger discount in utility for being selected less likely. (2) \textit{Numerical stability}: the speed at which the discount inflates should decrease as the staleness goes infinite.

\begin{figure}[t]
    \centering
  \begin{subfigure}[b]{0.40\columnwidth}
    \centering
    \includegraphics[width=\columnwidth]{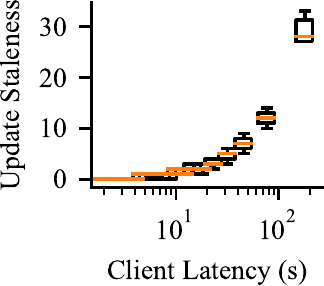}
    \caption{FEMNIST (Image)}
    \label{fig:async_staleness_femnist}
  \end{subfigure}
  \hspace{15px}
  \begin{subfigure}[b]{0.40\columnwidth}
    \centering
    \includegraphics[width=\columnwidth]{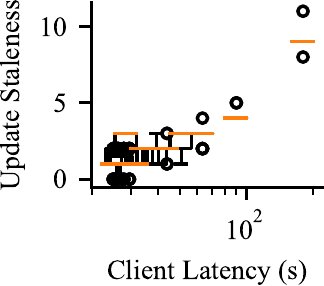}
    \caption{StackOverflow (LM)}
    \label{fig:async_staleness_stackoverflow}
  \end{subfigure}
  \caption[]{Clients' staleness varies slightly throughout the training regardless of their execution time.}
  \label{fig:async_staleness}
\end{figure}

It remains a question of how to estimate the actual staleness $\tau_i$ in implementation. By definition, its exact value is unknown until the client returns its update, which happens after the selection. While it may be precisely inferred through elaborate simulation of the federation, doing so is impractical due to the need for accurate knowledge of clients' speeds and abundant computing resources at the server. Given that Eq.~\ref{eq:Pisces_utility} is not designed to work with an accurate estimation on $\tau_i$, we advocate \textit{approximating it with historical records}. Specifically, Pisces lets $\tilde{\tau_i}$ be the moving average of the most recent $k$ actual values $\tau_{i, t-k+1}, \tau_{i, t-k+2}, \dots \tau_{i,t}$, i.e.,

\begin{equation}
    \tilde{\tau_i} = \frac{1}{k}\sum_{j=t-k+1}^{t} \tau_{i, j}.
    \label{eq:staleness_estimation}
\end{equation}

The intuition behind the approximation is that the staleness of a client's updates is usually \textit{stable} over time, given the stability of (1) clients' execution times and (2) the frequency of model aggregation. To exemplify, we study the staleness behaviors associated with the experiments mentioned in Sec.~\ref{sec:selection_concurrency}. We use 15 as the concurrency limit without loss of generality. As depicted in Fig.~\ref{fig:async_staleness}, during the training, the staleness of each client slightly fluctuates around the median, with the maximum range of individual values being $6$ and $3$ for FEMNIST and StackOverflow, respectively.
\section{Model Aggregation}\label{sec:aggregation}

While Pisces' optimizations in participant selection can mitigate stale computation, it cannot thoroughly prevent its occurrences. As for stale updates already generated, we operate at the phase of model aggregation to protect the global model from being arbitrarily impaired. In this section, we first start with the goal of bounded staleness and discuss the limitations of existing fixes (\cref{sec:aggregation_staleness}). We then elaborate on Pisces' principles on performing adaptive aggregation pace control for efficiently bounding clients' staleness (\cref{sec:aggregation_adaptive}).

\subsection{Bounded Staleness}\label{sec:aggregation_staleness}

As mentioned in Sec.~\ref{sec:selection_utility}, local updates with larger staleness empirically bring less gain in model convergence. This fact has a theoretical ground as stated below, consolidating the first-order goal of bounding staleness in aggregation.

\PHM{Why bounded staleness is desired?} The mainstream way to derive a convergence guarantee for asynchronous FL is based on the perturbed iterate framework~\cite{mania2017perturbed, nguyen2021federated}. In addition to assumptions that are commonly made in synchronous FL, it requires another assumption to hold true as follows:

\begin{assumption}
(Bounded Staleness) For all clients $i \in [N]$ and for each Pisces' server step, the staleness $\tau_i$ between the model version in which client $i$ uses to start local training, and the model version in which the corresponding update is used to modify the global model is not larger than $\tau_{\max}$. 
\label{as:staleness}
\end{assumption}

We here sketch the intuition on how bounded staleness helps convergence. By limiting the staleness of each client's updates all the time, we can \textit{limit the model divergence} between any version of the global model $w_t$ and the initial model that any corresponding contributor of $w_t$ uses. This implies that the contributors' gradients do not deviate much from each other. Consequently, each model aggregation attempt does take an effective step towards the training objective, and the training can thus terminate with finite times of aggregation. We provide more details in Sec.~\ref{sec:convergence}.

\begin{figure}[t]
  \centering
\begin{subfigure}[b]{0.98\columnwidth}
  \centering
  \includegraphics[width=\columnwidth]{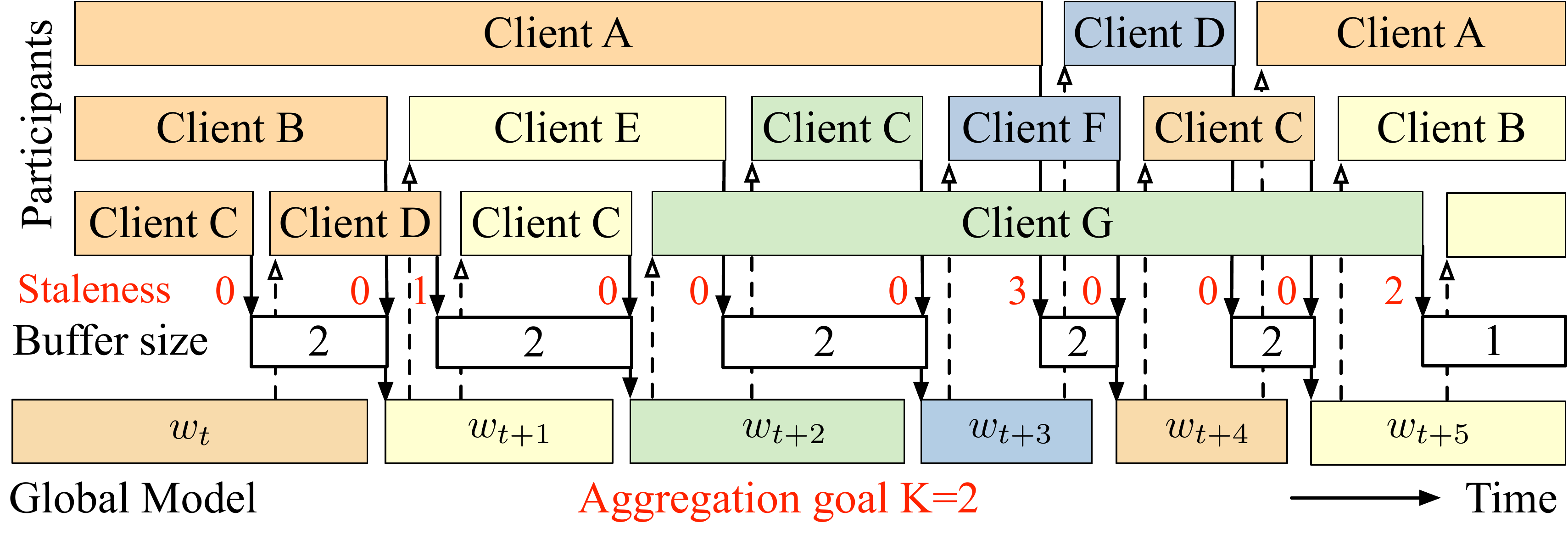}
  \caption{Buffered aggregation in FedBuff~\cite{nguyen2021federated}.}
  \label{fig:async_aggregation_ba}
\end{subfigure}
\begin{subfigure}[b]{0.98\columnwidth}
  \centering
  \includegraphics[width=\columnwidth]{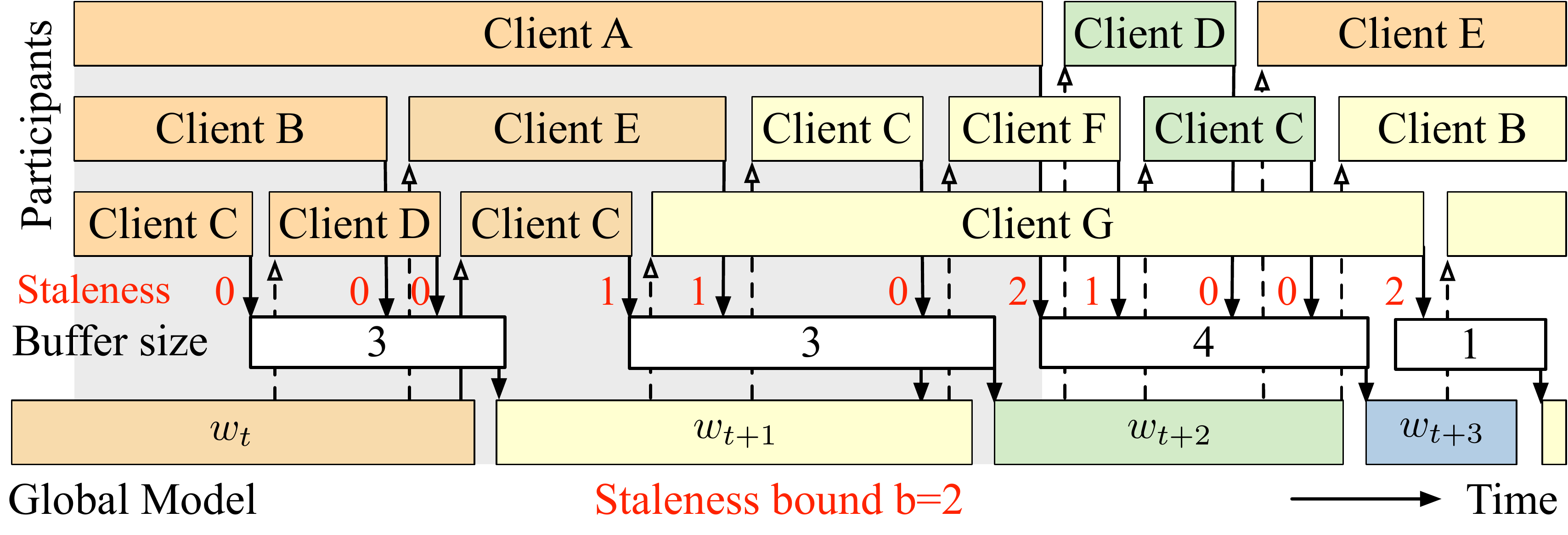}
  \caption{Pisces' adaptive pace control.}
  \label{fig:async_aggregation_Pisces}
\end{subfigure}
\caption[]{Two methods for steering aggregation pace.}
\label{fig:async_aggregation}
\end{figure}

\PHM{Does buffered aggregation suffice?} To control clients' staleness, the state-of-the-art design adopts \textit{buffered aggregation (BA)}. The server uses a buffer to store received local updates and only aggregates when the buffer size reaches a predefined aggregation goal $K > 1$, as illustrated in Fig.~\ref{fig:async_aggregation_ba}.

Nevertheless, due to the lack of explicit control, the maximum staleness across clients in BA can \textit{go unbounded}, to which extent depends on the heterogeneity degree of clients' speeds. For example, for the experiments reported by Fig.~\ref{fig:async_staleness}, the fastest client is $90\times$ faster than the slowest in FEMNIST and $7.2\times$ in StackOverflow. In this case, despite using the same aggregation goal ($K=6$), the maximum staleness values differ vastly (33 in FEMNIST and 11 in StackOverflow). This implies the need for \textit{manually tuning} the aggregation goal across different federation environments and learning tasks, which can be too expensive or even prohibited.

\subsection{Adaptive Pace Control}\label{sec:aggregation_adaptive}

To generally enforce bounded staleness, we develop an adaptive strategy for steering the pace of model aggregation.

\begin{alg}[t]
  \small
  \algrule
  \DontPrintSemicolon
  \SetNoFillComment
  \KwIn{Running clients $R$, target staleness bound $b > 0$, last aggregation time $t_{\tau_{j-1}}$, current time $T_l$, clients' profiled latencies $\{L_i\}_{i \in [n]}$}
  \KwOut{Aggregation decision for the loop step $l$}
  \algrule
  \SetKwInOut{Para}{Para.}
  \SetKwProg{Fn}{Function}{}{}
  \newcommand\mycommfont[1]{\small\rmfamily\textcolor{ACMPurple}{#1}}
  \SetCommentSty{mycommfont}
  \SetAlgoNoLine

  \SetKwFunction{ManagerIfAggregate}{ManagerIfAggregate}

  \Fn{\ManagerIfAggregate{}}{
      \tcc{Set the aggregation interval proportionate to the profiled latency of the slowest running client.}
      $L_{max} = \max_{i \in R} L_i$ \;
      $I = L_{max} / b$ \;
      \BlankLine

      \tcc{Aggregate if the interval currently ends.}
      \Return{$T_l - t_{\tau_{j-1}} > I$}\;
  }
  \algrule
  \caption{Adaptive aggregation pace control in Pisces.}
  \label{algo:agg}
\end{alg}

\PHM{How to be adaptive to complex dynamics?} In general, the distribution of the staleness values across clients is determined by the interplay between (1) the algorithms used (selection and aggregation), and (2) the dynamic environment (e.g., clients' speeds or data quality). While accounting for the whole population is overwhelming, we can reduce the problem to a narrowly scoped one that only considers an individual client. Our insight is that to enforce a staleness bound for all clients throughout the training, it suffices to bound the staleness of the slowest running client for each time unit. We thus propose to keep track of the running client with the \textit{largest profiled latency} and adjust the \textit{aggregation interval}, i.e., the time between two consecutive model aggregations, such that the aggregation pace can somehow match the slowest client's speed.

\PHM{How to adjust the interval?} We further develop the idea with a toy example as shown in Fig.~\ref{fig:async_aggregation_Pisces}. We focus on the time period where Client $A$ performs local training (highlighted with grey shadow). During this process, $A$ is consistently the slowest running client. Thus, by ensuring that $m \leq 2$ model updates take place during $A$'s training, we can guarantee the same thing for any other running clients within this period.

It only leaves off the question of when should these $m$ aggregations happen. Our intuition is that, arranging them evenly in terms of time can balance the numbers of contributors across different aggregation attempts in expectation, helping the global model \textit{evolve smoothly} in theory. Also, a (nearly) uniform distribution of the server's aggregation workload can sustain \textit{better scalability}.

\PHM{Latency-aware aggregation interval.} By generalizing the above idea to real deployments where the identity of the slowest running client changes over time, we obtain Pisces' principles on steering the model aggregation pace. In essence, Pisces periodically examines the necessity of aggregation in the control loop (\cref{sec:overview}). As outlined in Alg.~\ref{algo:agg}, for a loop step $l$ that begins at time $T_l$, Pisces first fetches the profiled latency $L_{max}$ of the slowest client that is currently running. It then determines the aggregation interval $I$ suitable to bound the client's staleness based on the above intuition. Finally, it computes the elapsed time $e = T_l - t_{\tau_{j-1}}$ since the last model aggregation happened. Only when $e$ is larger than the interval $I$ will Pisces suggest performing model aggregation in this step. In practice, clients' latencies can be profiled with historical records. As Sec.~\ref{sec:convergence} shows, assuming accurate latency predictions, Alg.~\ref{algo:agg} can realize bounded staleness.
\section{Convergence Analysis}\label{sec:convergence}

In this section, we first prove that Pisces guarantees bounded staleness. We then present Pisces' convergence guarantee.

\PHM{Notation.} We use the following notation: $T_l$ denotes the start time of a loop step $l$, $I_l$ denotes the aggregation interval generated by Alg.~\ref{algo:agg} at step $l$, $n$ denotes the total number of clients, $b$ denotes the target staleness bound used by Pisces, $\nabla F_i(w)$ denotes the gradient of model $w$ with respect to the loss of client $i$'s data, $f(w) = \frac{1}{n} \sum_{i=1}^n p_i F_i(w)$ denotes the global learning objective with $p_i > 0$ weighting each client's contribution, $f^*$ denotes the optimum of $f(w)$, $g_i(w; \xi_i)$ denotes the stochastic gradient on client $i$ with randomness $\xi_i$, and $Q$ denotes the number of steps in local SGD.

\PHM{Why does Pisces achieve bounded staleness?} We first state a useful lemma that globally holds for each loop step.

\begin{lemma}
  For any loop step $l$ where model aggregation happens, there must be no model aggregation happening during the time range $[T_l - I_l, T_l)$.
  \label{lem:interval}
\end{lemma}

\begin{proof}
Assume that a model aggregation happens at time $t \in [T_l - I_l, T_l)$, the elapsed time since this aggregation starts is then $T_l - t \leq I_l$. By the design of Alg.~\ref{algo:agg}, there should be no aggregation in step $l$, which leads to a contradiction.
\end{proof}

We are then able to derive the bounded staleness property.

\begin{theorem}
Executing Alg.~\ref{algo:agg} for all loop steps $l=1, \cdots, L$, the maximum number of model aggregations happening during any training process of any client $i\in[n]$ is no more than $b$, provided that the profiled latencies $\{L_i\}_{i\in[n]}$ are accurate.
\end{theorem}

\begin{proof}
  Consider a training process of client $i$ that lasts for $L_i$. In the interim, assume that there are $m \in \mathbb{N}$ model aggregations performed by the server, each of which happens in loop step $l_k$ where $k \in [1, m]$.

  Now, applying Lemma~\ref{lem:interval}, the duration between the 1st and the $m$-th aggregation has a lower bound $T_{l_m} - T_{l_1} > \sum_{j=2}^{m} I_{l_j}$. By definition, $I_{l_j} = L_{max, l_j} / b$ where $L_{max, l_j}$ is the end-to-end latency of the slowest client running at step $l_j$. Given that $L_{max, l_j} \geq L_i$, we further have $T_{l_m} - T_{l_1} >  (m-1) L_i/b$.

  On the other hand, we also have $L_i \geq T_{l_m} - T_{l_1}$, as the two aggregations all happen in this training process. Combining the two inequalities yields $1 > (m-1)/b$, i.e., $m < b$.
\end{proof}

\PHM{What is the convergence guarantee?} We additionally make the following assumptions which are commonly made in analyzing FL algorithms~\cite{stich2018local, yu2019parallel, li2019convergence, reddi2020adaptive}.

\begin{assumption}
  (Unbiasedness of client stochastic gradient) $\mathbb{E}_{\zeta_i}[g_i(w ; \zeta_i))] = \nabla F_i(w)$.
  \label{as:unbiased}
\end{assumption}
  
\begin{assumption}
  (Bounded local and global variance) for all clients $i \in [1, n]$, $\mathbb{E}_{\zeta_i|i}[\norm{g_i(w; \zeta_i) - \nabla F_i(w)}^2] \leq \sigma^2_{\ell},$
  and 
  $\frac{1}{n} \sum_{i=1}^n \norm{\nabla F_i(w) - \nabla f(w)}^2 \leq \sigma^2_{g}$.
  \label{as:bounded_var}
\end{assumption}
  
\begin{assumption}
  (Bounded gradient) $\norm{\nabla F_i}^2 \leq G$, $i \in [1, n]$.
  \label{as:bounded_grad}
\end{assumption}

\begin{assumption}
  (Lipschitz gradient) for all client $i \in [1, n]$, the gradient is $L$-smooth $\norm{\nabla F_i(w) - \nabla F_i(w')}^2 \leq L \norm{w - w'}^2.$
  \label{as:lipz}
\end{assumption}

Given the five assumptions, by substituting the use of maximum delay $\tau_{max, K}$ with our enforced staleness bound $b$ and using a constant server learning rate $\eta_g = 1$ (as we execute Federated Averaging~\cite{mcmahan2017communication}) in \cite{nguyen2021federated}'s proof, we can derive the convergence guarantee for Pisces as follows.

\begin{theorem}
  Let $\eta^{(q)}_{\ell}$ be the local learning rate of client SGD in the $q$-th step, and define $\alpha(Q):=\sum_{q=0}^{Q-1} \eta^{(q)}_{\ell}$, $\beta(Q):=\sum_{q=0}^{Q-1} (\eta^{(q)}_{\ell})^2$. Choosing $\eta^{(q)}_{\ell} Q \leq \frac{1}{L}$ for all local steps $q=0,\cdots,Q-1$, the global model iterates in Pisces achieves the following ergodic convergence rate
  \begin{equation}
      \begin{aligned}
          \frac{1}{T} \sum_{t=0}^{T-1} \norm{\nabla f(w^t)}^2 
           &\leq \frac{2 \Big(f(w^0) - f^* \Big)}{\alpha(Q) T} +\frac{L}{2}\frac{\beta(Q) }{ \alpha(Q)}   \sigma^2_{\ell} \\ + 3 L^2 Q  & \beta(Q)  \Big(b^2 + 1 \Big)  \Big(\sigma^2_{\ell} + \sigma^2_g +  G\Big).
      \end{aligned}
  \end{equation}
  \label{thm:convergence}
\end{theorem}
\section{Implementation}\label{sec:implementation}

We implement Pisces atop Plato~\cite{tlsystem2021plato}, an open-source framework for scalable FL research, with 2.1k lines of code. 

\PHM{Platform.} Plato abstracts away the underlying ML driver with handy APIs, with which we can seamlessly port Pisces to infrastructures including PyTorch \cite{paszke2019pytorch}, Tensorflow \cite{abadi2016tensorflow}, and MindSpore \cite{mindspore}. On the other hand, at the time of engineering, Plato did not support asynchronous FL and participant selection. Thus, we first implement the coordinator's control loop and client manager for enhanced functionality. For fairness of model testing overhead across different synchronization modes, we take the testing logic from the main loop to a concurrent background process. Further, to emulate arbitrary distributions of clients' speeds given finite hardware specifications, we instrument Plato to simulate client latencies by controlling exactly when a received local update is "visible" to the FL protocol. This enables us to control clients' latencies in a fine-grained manner, as shown in Sec.~\ref{sec:evaluation_method}.
\section{Evaluation}\label{sec:evaluation}

We evaluate Pisces's effectiveness in various FL training tasks. The highlights of our evaluation are listed as follows.

\begin{enumerate}
    \item Pisces \textbf{outperforms Oort and Fedbuff} by 1.2$-$2.0$\times$ in time-to-accuracy performance.. It gains high training efficiency by automating participant selection and model aggregation in a principled way (\cref{sec:evaluation_end2end}).
    \item Pisces is superior to its counterparts over different \textbf{participation scales} and \textbf{degrees of system heterogeneity} across clients. Its efficiency is also shown to be \textbf{insensitive} to choice of the hyperparameters (\cref{sec:evaluation_sensitivity}).
\end{enumerate}

\subsection{Methodology}\label{sec:evaluation_method}

\PHM{Datasets and models.} We run two categories of applications with four real-world datasets of diverse scales.

\begin{itemize}
    \item \textit{Image Classification:} the CIFAR-10 dataset \cite{krizhevsky2009learning} with 60k colored images in 10 classes, the MNIST dataset \cite{deng2012mnist} with 60k greyscale images in 10 classes, and a larger dataset, FEMNIST \cite{caldas2018leaf}, with 805k greyscale images in 62 categories collected from 3.5k data owners. We train LeNet5 \cite{lecun1989backpropagation} to classify the images in MNIST and FEMNIST and use ResNet-18 \cite{he2016deep} for CIFAR-10.
    \item \textit{Language Modeling:} a large-scale StackOverflow \cite{stackoverflow} dataset contributed by 0.3 million clients. We train Albert \cite{lan2019albert} over it for next-word prediction.
\end{itemize}

\PHM{Experimental Setup.} We use a 20-node cluster to emulate 200 clients, where each node is an AWS EC2 \texttt{c5.4xlarge} instance (16 vCPUs and 32 GB memory). We launch another \texttt{c5.4xlarge} instance for the server. Unless otherwise stated, clients' system heterogeneity and data heterogeneity are configured independently. For system heterogeneity, we introduce barriers to the execution of clients (\cref{sec:implementation}) so that their end-to-end latencies follow the Zipf distribution parameterized by $a = 1.2$ (moderately skewed).\footnote{The end-to-end latency of the $i$th slowest client is proportional to $i^{-a}$, where $a$ is the parameter of the distribution.} To introduce data heterogeneity, we resort to either of the two solutions:

\begin{itemize}
    \item \textit{Synthetic Datasets.} For datasets that are derived from the ones used in conventional ML (MNIST and CIFAR-10), we apply latent Dirichlet allocation (LDA) over data labels for each client as in the literature \cite{hsu2019measuring, reddi2020adaptive, al2020federated, acar2021federated}. We set the concentration parameters to be a vector of 1.0's that corresponds to highly non-IID label distributions across clients where each of them biases towards different subsets of all the labels.
    \item \textit{Realistic Datasets.} For datasets collected in real distributed scenarios (FEMNIST and StackOverflow), as they have been partitioned by the original data owners, we directly allocate one partition to a client to preserve the native non-IID properties.
\end{itemize}

\begin{table}[t]
  \begin{center}
      \caption{Summary of the  hyperparameters.}
  \label{tab:parameters}
  \resizebox{\columnwidth}{!}{%
  \begin{tabular}{lcccc}
      \toprule
      Parameters & MNIST & FEMNIST & CIFAR-10 & StackOverflow \\
      \midrule
      Local epochs & 5 & 5 & 1 & 2 \\
      Batch size & 32 & 32 & 128 & 20 \\
      Learning rate & 0.01 & 0.01 & 0.01 & 0.00008 \\
      Momentum & 0.9 & 0.9 & 0.9 & 0.9 \\
      Weight decay & 0 & 0 & 0.0001 & 0.0001 \\
      \bottomrule
  \end{tabular}
  }
  \end{center}
\end{table}

\PHM{Hyperparameters.} The concurrency limit is $C=20$ (\cref{sec:selection_concurrency}; shared with Oort and FedBuff), and staleness penalty factor is $\beta=0.5$ (\cref{sec:selection_utility}). \textbf{The target staleness bound $b$ (\cref{sec:aggregation_adaptive}) always equates $C$}. We use SGD except for StackOverflow where Adam \cite{kingma2014adam} is used. Other configurations are listed at Table~\ref{tab:parameters}.

\PHM{Baseline Algorithms.} We evaluate Pisces against Oort \cite{lai2021oort}, the state-of-the-art optimized solution for synchronous FL, and FedBuff \cite{nguyen2021federated}, the cutting-edge asynchronous FL  algorithm. We use the default set of hyperparameters for Oort. The aggregation goal of FedBuff is set to 20\% of the concurrency limit, according to the authors' suggestions.

\PHM{Metrics.} We primarily care about the \textit{elapsed time} taken to reach the target accuracy. To capture the training dynamics, we also record the \textit{number of involvements} of each client and the \textit{number of aggregations} performed by the server.

\begin{table*}[t]
  \begin{center}
      \caption{Summary of Pisces' improvements on the time-to-accuracy performance.}
  \label{tab:time_to_acc}
  \begin{tabular}{cccc>{\columncolor{yellow!30}}ccc}
      \toprule
      \multirow{2.5}{*}{Task} & \multirow{2.5}{*}{Dataset} & \multirow{2.5}{*}{\shortstack[c]{Target\\Accuracy}} & \multirow{2.5}{*}{Model} & \multicolumn{3}{c}{Time-to-Accuracy} \\ \cmidrule(l){5-7}
      & & & & Pisces & Oort \cite{lai2021oort} & FedBuff \cite{nguyen2021federated} \\
      \midrule
      \multirow{3}{*}{\shortstack[c]{Image\\Classification}} & MNIST \cite{deng2012mnist} & 97.8\% & LeNet-5 \cite{lecun1989backpropagation} & 6.2min & 12.8min (2.0$\times$) & 7.6min (1.2$\times$) \\

      & FEMNIST \cite{caldas2018leaf} & 60.0\% & LeNet-5 & 8.9min & 16.0min (1.8$\times$) & 12.6min (1.4$\times$) \\

      & CIFAR-10 \cite{krizhevsky2009learning} & 65.1\% & ResNet-18 \cite{he2016deep} & 24.5min & 40.3min (1.6$\times$) & 26.5min (1.1$\times$) \\

      \midrule
      Language Modeling & StackOverflow \cite{stackoverflow} & 800 perplexity & Albert \cite{lan2019albert} & 25.0min & 48.4min (1.9$\times$) & 47.2min (1.9$\times$) \\
      \bottomrule
  \end{tabular}
  \end{center}
\end{table*}

\subsection{End-to-End Performance}\label{sec:evaluation_end2end}

We start with Pisces' time-to-accuracy performance. We first highlight its improvements over Oort and FedBuff. We further break down the source of efficiency in Pisces.

\PHM{Pisces improves time-to-accuracy performance.} Table~\ref{tab:time_to_acc} summarizes the speedups of Pisces over the two baselines, where Pisces reaches the target $1.2$-$2.0\times$ faster on the three image classification tasks. A consistent speedup of $1.9\times$ can be observed in language modeling. To understand the source of such improvements, we first compare Pisces against Oort. As depicted in Fig.~\ref{fig:over_oort}, the accumulated number of model aggregations in Pisces is constantly larger than that in Oort (and so does that in FedBuff). This demonstrates the advantages of removing the synchronization barriers.

\begin{figure}[t]
  \centering
  \begin{subfigure}[b]{0.75\columnwidth}
    \centering
    \includegraphics[width=\columnwidth]{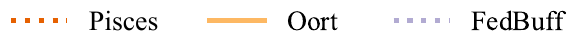}
  \end{subfigure}
  \begin{subfigure}[b]{0.46\columnwidth}
    \centering
    \includegraphics[width=\columnwidth]{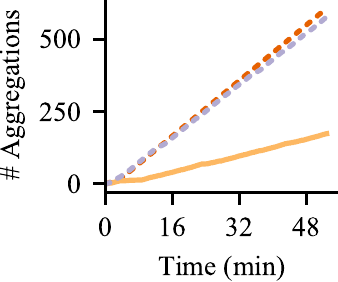}
    \caption{FEMNIST (Image)}
    \label{fig:time_to_acc_femnist}
  \end{subfigure}
  \hspace{3.0px}
  \begin{subfigure}[b]{0.45\columnwidth}
    \centering
    \includegraphics[width=\columnwidth]{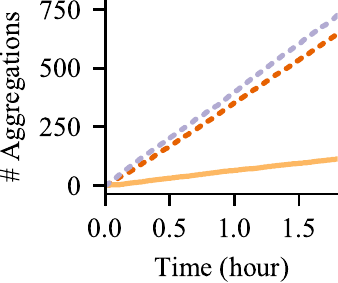}
    \caption{StackOverflow (LM)}
    \label{fig:time_to_acc_stackoverflow}
\end{subfigure}
\caption{Pisces performs model aggregation more frequently than Oort for being asynchronous.}
\label{fig:over_oort}
\end{figure}

On the other hand, albeit with comparable update frequency as in Pisces, FedBuff improves over Oort to a milder degree. The key downside of FedBuff lies in its random selection method. As shown in Fig.~\ref{fig:over_fedbuff}, Pisces prefers clients with large datasets, while FedBuff shows barely any difference in the interests of clients. Given the intuition that clients with larger datasets have higher potential to improve the global model, Pisces makes more efficient use of concurrency quotas than FedBuff does. Oort also differentiates clients by data quality, though to a more moderate extent as it has to reconcile for clients' speeds.

\begin{figure}[t]
  \centering
  \begin{subfigure}[b]{0.47\columnwidth}
    \centering
    \includegraphics[width=\columnwidth]{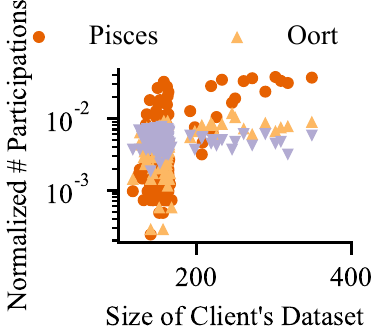}
    \caption{FEMNIST (Image)}
    \label{fig:time_to_acc_femnist}
\end{subfigure}
  \centering
  \begin{subfigure}[b]{0.45\columnwidth}
    \centering
    \includegraphics[width=\columnwidth]{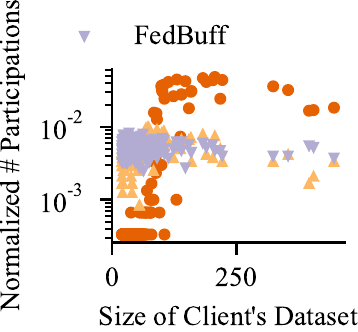}
    \caption{StackOverflow (LM)}
    \label{fig:time_to_acc_stackoverflow}
\end{subfigure}
\caption{Pisces selects informative clients more frequently than FedBuff with principled selection.}
\label{fig:over_fedbuff}
\end{figure}

\PHM{Pisces exhibits stable convergence behaviors.} Fig.~\ref{fig:time_to_acc} further plots the learning curves of different protocols. First, when reaching the corresponding time limit, Pisces achieves 0.2\%, 2.9\%, and 0.8\% higher accuracy on MNIST, FEMNIST, and CIFAR-10, respectively, and 23 lower perplexity on StackOverflow, as compared to Oort. In other words, apart from theoretical convergence guarantees on convergence, Pisces is also empirically shown to achieve comparative final model quality with that in Oort (synchronous FL).

\begin{figure}[t]
  \centering
  \begin{subfigure}[b]{0.75\columnwidth}
    \centering
    \includegraphics[width=\columnwidth]{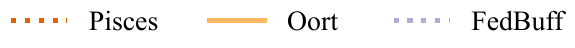}
  \end{subfigure}
  \begin{subfigure}[b]{0.45\columnwidth}
    \centering
    \includegraphics[width=\columnwidth]{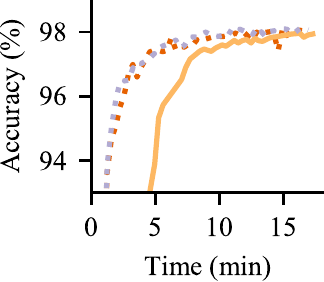}
    \caption{MNIST (Image)}
    \label{fig:time_to_acc_mnist}
  \end{subfigure}
    \hspace{3.0px}
    \begin{subfigure}[b]{0.45\columnwidth}
      \centering
      \includegraphics[width=\columnwidth]{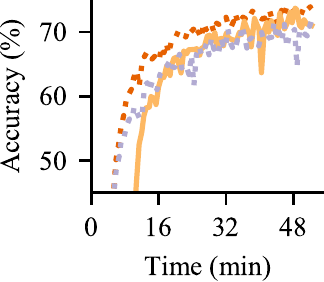}
      \caption{FEMNIST (Image)}
      \label{fig:time_to_acc_femnist}
  \end{subfigure}
  \begin{subfigure}[b]{0.45\columnwidth}
    \centering
    \includegraphics[width=\columnwidth]{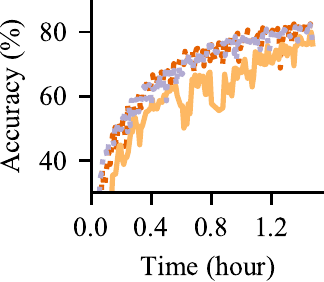}
    \caption{CIFAR-10 (Image)}
    \label{fig:time_to_acc_cifar}
  \end{subfigure}
  \hspace{3.0px}
  \begin{subfigure}[b]{0.45\columnwidth}
    \centering
    \includegraphics[width=\columnwidth]{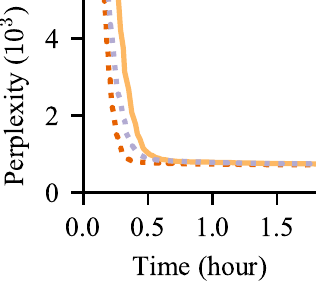}
    \caption{StackOverflow (LM)}
    \label{fig:time_to_acc_stackoverflow}
  \end{subfigure}
\caption{Elaboration on the convergence behaviors.}
\label{fig:time_to_acc}
\end{figure}

Moreover, in all evaluated cases, Pisces' model quality evolves more stably than that in Oort, especially in the middle and late stages. Pisces' stability advantage probably stems from the noise in local updates induced by moderately stale computation, which is analogous to the one introduced to avoid overfitting in traditional ML \cite{srivastava2014dropout}. This conjecture complies with the fact that FedBuff also exhibits stable convergence, though sometimes to an inferior model quality (3.1\% and 0.4\% lower accuracy than Pisces on FEMNIST and CIFAR-10, respectively, and 24 higher perplexity on StackOverflow).

\PHM{Pisces' optimizations on participant selection are effective.} To examine which part of Pisces' participant selection optimizations to accredit with, we break it down and formulate two variants: (i) \textit{Pisces w/o slt.}: we disable Pisces' selection strategies and sample randomly instead; (ii) \textit{Pisces w/o stale.}: while we still select clients based on their data quality, we ignore the impact that clients' staleness has on clients' utilities, as if the second term in Eq.~\ref{eq:Pisces_utility} is consistently set to be one. Fig.~\ref{fig:breakdown_slt} reports the comparison of complete Pisces with these two variants.

\begin{figure}[t]
  \centering
  \begin{subfigure}[b]{1.0\columnwidth}
    \centering
    \includegraphics[width=\columnwidth]{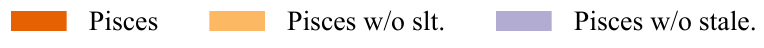}
  \end{subfigure}
  \begin{subfigure}[b]{0.39\columnwidth}
    \centering
    \includegraphics[width=\columnwidth]{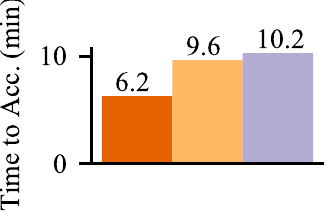}
    \caption{MNIST (Image)}
    \label{fig:breakdown_slt_mnist}
\end{subfigure}
\hspace{18px}
  \centering
  \begin{subfigure}[b]{0.40\columnwidth}
    \centering
    \includegraphics[width=\columnwidth]{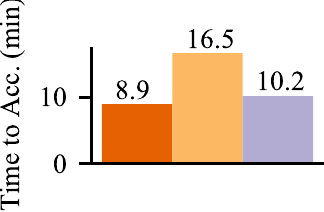}
    \caption{FEMNIST (Image)}
    \label{fig:breakdown_slt_femnist}
\end{subfigure}
\begin{subfigure}[b]{0.40\columnwidth}
  \centering
  \includegraphics[width=\columnwidth]{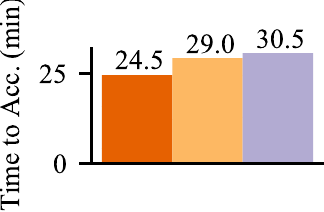}
  \caption{CIFAR-10 (Image)}
  \label{fig:breakdown_slt_cifar}
\end{subfigure}
\hspace{18px}
\begin{subfigure}[b]{0.40\columnwidth}
  \centering
  \includegraphics[width=\columnwidth]{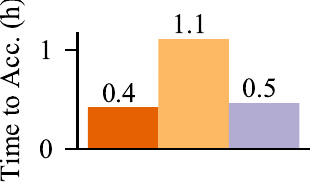}
  \caption{StackOverflow (LM)}
  \label{fig:breakdown_slt_stackoverflow}
\end{subfigure}
\caption{Breakdown of participant selection designs.}
\label{fig:breakdown_slt}
\end{figure}

As Pisces selects clients with high data quality and low potential of inducing stale computation, it improves the time-to-accuracy over \textit{Pisces w/o slt.} by $1.1$-$2.7\times$. Further, the criteria on data quality are more critical than those on staleness dynamics, as Pisces is shown to enhance the performance of \textit{Pisces w/o stale.} less significantly ($1.1$-$1.6\times$). Understandably, with adaptive pace control in model aggregation, the pressure of avoiding stale computation in participant selection is partially released. Still, the combined considerations of the two factors yield the best efficiency.

\PHM{Pisces' model optimizations on aggregation adapt to various settings.} To understand the benefits of adaptive pace control, we also compare Pisces against another variant: (i) \textit{Pisces w/o adp.}: we disable Pisces' adaptive pace control and instead resort to buffered aggregation (\cref{sec:aggregation_staleness}) with various choices of the aggregation goal $K$: 5\%, 10\% and 40\% of the concurrency limit $C$, representing gradually decreasing aggregation frequencies. As mentioned in Sec.~\ref{sec:aggregation_adaptive}, clients' staleness dynamics partly depend on the distributions of clients' speeds. We thus also vary the skewness of client latency distribution by using different $a$'s in the Zipf distribution: 1.2 (moderate), 1.6 (high), and 2.0 (heavy).

\begin{figure}[t]
  \centering
  \begin{subfigure}[b]{1.0\columnwidth}
    \centering
    \includegraphics[width=\columnwidth]{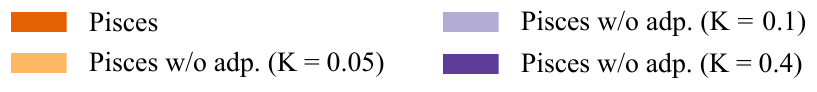}
  \end{subfigure}
  \begin{subfigure}[b]{1.0\columnwidth}
    \centering
    \includegraphics[width=\columnwidth]{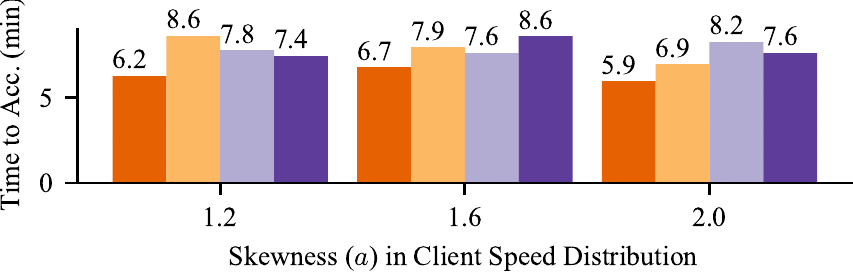}
    \caption{MNIST (Image)}
    \label{fig:breakdown_agg_mnist}
  \end{subfigure}
  \begin{subfigure}[b]{1.0\columnwidth}
    \centering
    \includegraphics[width=\columnwidth]{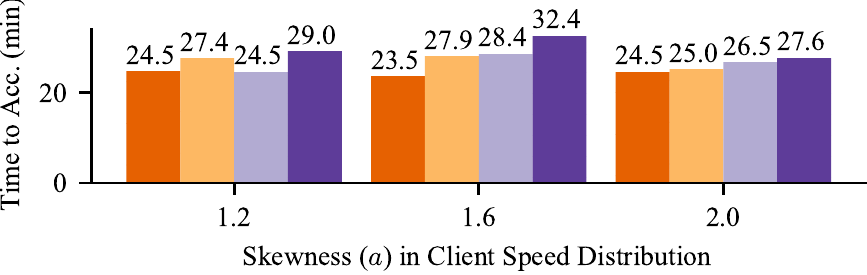}
    \caption{CIFAR-10 (Image)}
    \label{fig:breakdown_agg_cifar10}
  \end{subfigure}
\caption{Aggregation adapts to speed heterogeneity.}

\label{fig:breakdown_agg}
\end{figure}

As shown in Fig.~\ref{fig:breakdown_agg}, Pisces consistently improves over \textit{Pisces w/o adp.} with different aggregation goals by up to $1.4\times$ for both MNIST and CIFAR-10, respectively. Further, across various degrees of client speed heterogeneity, \textit{Pisces w/o adp.} exhibits unstable efficiency as sticking to any of $K$ does not yield the best performance for all cases. Thus, buffered aggregation relies on the manual tuning of $K$ for unleashing the maximum potential, while Pisces's adaptive pace control is more preferable with (1) a deterministic bound of clients' staleness for theoretically guaranteed efficiency and (2) one parameter $b$ which does not require to tune.

\subsection{Sensitivity Analysis}\label{sec:evaluation_sensitivity}

We also examine Pisces' effectiveness across various environments and configurations. All the results are based on the same accuracy targets mentioned in Sec.~\ref{sec:evaluation_end2end}.

\begin{figure}[t]
  \centering
  \begin{subfigure}[b]{0.8\columnwidth}
    \centering
    \includegraphics[width=\columnwidth]{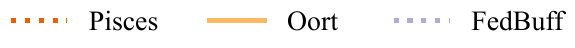}
  \end{subfigure}
  \begin{subfigure}[b]{0.47\columnwidth}
    \centering
    \includegraphics[width=\columnwidth]{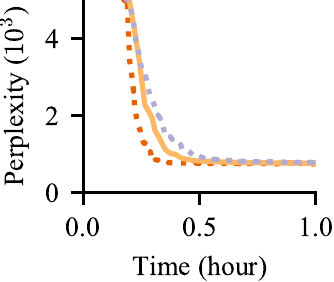}
    \caption{StackOverflow ($C=100$)}
    \label{fig:scale_stackoverflow_100}
\end{subfigure}
\hspace{10px}
\begin{subfigure}[b]{0.47\columnwidth}
  \centering
  \includegraphics[width=\columnwidth]{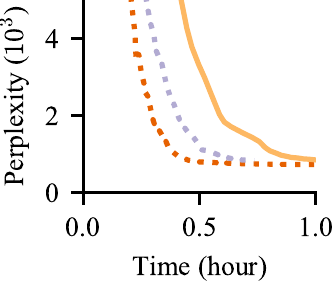}
  \caption{StackOverflow ($C=400$)}
  \label{fig:scale_stackoverflow_400}
\end{subfigure}
\caption{Pisces outperforms in various scales.}

\label{fig:scale}
\end{figure}

\PHM{Impact of participation scale.} In addition to a pool of 200 clients with the concurrency limit being 20 ($N$=200 with $C$=20 in \cref{sec:evaluation_end2end}), we evaluate Pisces on two more scales of participation: $N$=100 with $C$=10, and $N$=400 with $C$=40. As shown in Fig.~\ref{fig:scale}, as the number of clients scales up, Pisces outperforms Oort (resp. FedBuff) in time-to-perplexity by $1.7\times$ (resp. 2.1$\times$), 1.9$\times$ (resp. 1.9$\times$), and 2.4$\times$ (resp. 1.6$\times$) for $N$=100, $N$=200, and $N$=400, respectively. The key reason for Pisces' stable performance benefits is that its algorithmic designs is agnostic to the population size. Further, Pisces enhances its scalability by (1) enforcing a concurrency limit (\cref{sec:selection_concurrency}), and (2) adaptively balancing the aggregation workload (\cref{sec:aggregation_adaptive}).

\begin{figure}[t]
  \centering
  \begin{subfigure}[b]{0.47\columnwidth}
    \centering
    \includegraphics[width=\columnwidth]{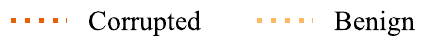}
  \end{subfigure}
  \begin{subfigure}[b]{0.51\columnwidth}
    \centering
    \includegraphics[width=\columnwidth]{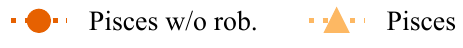}
  \end{subfigure}
  \begin{subfigure}[b]{0.39\columnwidth}
    \centering
    \includegraphics[width=\columnwidth]{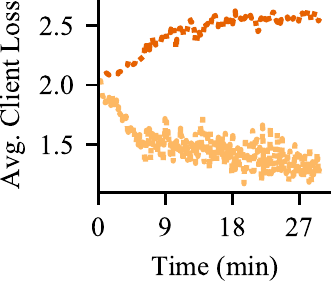}
    \caption{Loss Dynamics}
    \label{fig:corruption_loss}
\end{subfigure}
\hspace{18px}
\begin{subfigure}[b]{0.39\columnwidth}
  \centering
  \includegraphics[width=\columnwidth]{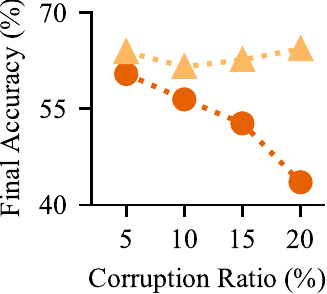}
  \caption{Robustness}
  \label{fig:corruption_acc}
\end{subfigure}
\caption{Pisces is robust against corrupted clients.}
\label{fig:corruption}
\end{figure}

\PHM{Impact of training loss outliers.} To validate the robustness to training loss outliers, we insert manual corruption into the FEMNIST dataset. We follow the literature on adversarial attacks~\cite{fang2020local, lai2021oort} to randomly flip all the labels of a subset of clients. This leads to consistently higher losses from corrupted clients than the majority, as exemplified in Fig.~\ref{fig:corruption_loss} where $5\%$ clients are corrupted. We compare Pisces against \textit{Pisces w/o rob.}, a variant of Pisces where anomalies are not identified and precluded. As shown in Fig.~\ref{fig:corruption_acc}, Pisces outperforms it in the final accuracy for accurately precluding outlier clients (\cref{sec:selection_utility}) across various scales of corruption.

\PHM{Impact of staleness penalty factor.} We next examine Pisces under different choices of the staleness penalty factor, $\beta$, which is introduced to prevent stale computation in participant selection (\cref{sec:selection_utility}). We set $\beta=0.2$ and $\beta=0.8$ in addition to the primary choice where $\beta=0.5$, where a large $\beta$ can be interpreted as a stronger motivation for Pisces to avoid stale computation. As depicted in Fig.~\ref{fig:staleness}, while different FL tasks may have different optimal choices of $\beta$ (e.g., around 0.5 for FEMNIST and 0.2 for StackOverflow), Pisces still improves performance across different uses of $\beta$.

\begin{figure}[t]
  \centering
  \begin{subfigure}[b]{0.92\columnwidth}
    \centering
    \includegraphics[width=\columnwidth]{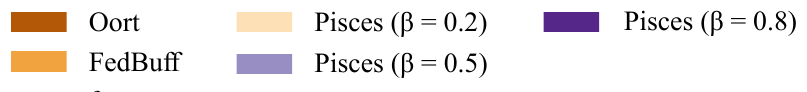}
  \end{subfigure}
  \begin{subfigure}[b]{0.39\columnwidth}
    \centering
    \includegraphics[width=\columnwidth]{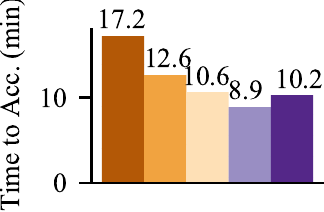}
    \caption{FEMNIST (Image)}
    \label{fig:staleness_femnist}
\end{subfigure}
\hspace{18px}
\begin{subfigure}[b]{0.39\columnwidth}
  \centering
  \includegraphics[width=\columnwidth]{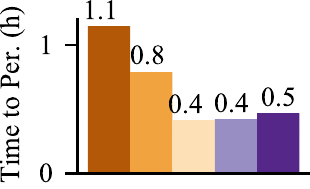}

  \caption{StackOverflow (LM)}
  \label{fig:staleness_stackoverflow}
\end{subfigure}
\caption{Selection with various staleness penalties.}
\label{fig:staleness}
\end{figure}
\section{Discussion}\label{sec:discussion}

\PHB{Generic FL training framework.} The abstraction that Pisces' architecture embodies, i.e., the three components (Fig.~\ref{fig:Pisces}) and the interactions (Fig.~\ref{fig:loop}), can be used to instantiate a wide range of FL protocols. Synchronous FL also fits in this abstraction by ensuring that (1) aggregation starts when no client is running, and (2) selection starts right after an aggregation. Given the expressiveness, we present it as a \textit{generic FL training framework} for helping FL developers compare various protocols more clearly and fairly.

\PHM{Privacy compliance.} During a client's participation in Pisces, the possible sources of information leakage are two-fold--the average training loss and local update that it reports to the server. For the current release of the plaintext training loss, we do not enlarge the attack surface compared to synchronous FL deployments~\cite{yang2018applied, hartmann2019federated,lai2021oort}, and thus can also adopt the same techniques as theirs to enhance privacy.

To prevent the local update from leaking sensitive information, Pisces can resort to differential privacy (DP)~\cite{dwork2006calibrating}, a rigorous measure of information disclosure about clients' participation in FL. Traditional DP methods~\cite{kairouz2021distributed, agarwal2021skellam} require precise control over which clients to contribute to a server update, which is not the case in asynchronous FL. However, Pisces is compatible with DP-FTRL~\cite{kairouz2021practical}, a newly emerged DP solution for privately releasing the prefix sum of a data stream. While we focus on improving FL efficiency and scalability, we plan to integrate Pisces with DP as future work. 
\section{Related Work}\label{sec:related}

\PHB{Federated learning.} Many asynchronous FL efforts focus on regularizing local optimization for accommodating data heterogeneity (FedAsync \cite{xie2020asynchronous}, ASO-Fed \cite{chen2020asynchronous}) and steering the pace of aggregation for handling staleness (HySync \cite{shi2020hysync}). However, they are designed for and evaluated on full concurrency, which leaves off suboptimal utilization and scalability (\cref{sec:selection_concurrency}). For designs considering fractional participation, they do not optimize either model aggregation (FedAR \cite{imteaj2020fedar}, \cite{chen2021towards}, \cite{hu2021device}) or participant selection (TEA-fed \cite{zhou2021tea}, FedBuff~\cite{nguyen2021federated}). Instead, Pisces comprehensively studies both design knobs for reaping the most benefits of asynchrony \cite{xu2021asynchronous}.

\PHM{Traditional machine learning.}  Despite sharing some insights with FL, e.g., modulating the learning rate \cite{zhang2016staleness} or bounding progress differences \cite{ho2013more} for staleness mitigation, asynchronous traditional ML does not expect data heterogeneity or a massive population of remote workers. Its solutions are thus confined to shared memory systems \cite{niu2011hogwild} or HPC clusters \cite{dean2012large, ho2013more, chilimbi2014project, zhang2016staleness, fan2018adaptive} and do not apply to FL.
\section{Conclusion}\label{conclusion}

While the trade-off between clients' speeds and data quality is knotty to navigate in synchronous FL, resorting to asynchronous FL requires retaining resource efficiency and avoiding stale computation. In this paper, we have presented practical principles for automating participant selection and model aggregation with Pisces. Compared to prior arts, Pisces features noticeable speedups, robustness, and flexibility.


\bibliographystyle{ACM-Reference-Format}
\bibliography{main.bib}


\end{document}